\documentclass[a4paper,UKenglish,cleveref, autoref, thm-restate]{lipics-v2021}

\title{Beating Trivial Time for Tricky Triangle Tasks} 

\author{Neha Pant}{MIT}{nehapant@mit.edu}{}{}
\author{Ryan Williams}{MIT}{rrw@mit.edu}{}{}

\authorrunning{N. Pant and R. Williams} 

\Copyright{Neha Pant and Ryan Williams} 



\relatedversion{} 

\funding{Parts of this work were completed while R.W. was visiting the Institute for Advanced Study, Princeton, NJ. This material is based upon work supported by the National Science Foundation under grants DMS-2424441 (at IAS) and CCF-2420092 (at MIT).} 

\nolinenumbers 

\EventEditors{Michal Kouck\'{y} and Daniela Petri\c{s}an}
\EventNoEds{2}
\EventLongTitle{51st International Symposium on Mathematical Foundations of Computer Science (MFCS 2026)}
\EventShortTitle{MFCS 2026}
\EventAcronym{MFCS}
\EventYear{2026}
\EventDate{August 24--28, 2026}
\EventLocation{Paris, France}
\EventLogo{}
\SeriesVolume{386}
\ArticleNo{8}

\usepackage{amsmath, amssymb,bbm,amsthm}
\usepackage{thm-restate,thmtools}

\newcommand{\eps}{\varepsilon}

\newcommand{\poly}{{\sf{poly}}}
\newcommand{\E}{\mathbb{E}}

\newcommand{\ac}{{\mathsf{AC}}}
\newcommand{\AC}{\ac}
\newcommand{\calH}{\mathcal{H}}
\newcommand{\nc}{{\mathsf{NC}}}

\newenvironment{reminder}[1]{\medskip
	\noindent {\sf{\textbf{Reminder of #1.}}}}{\smallskip}

\newif\iffullversion

\fullversiontrue    

\begin{document}

\maketitle

\begin{abstract}
For several well-studied triangle detection problems in the literature, the trivial enumeration algorithms are known to be optimal (up to the exponent) assuming popular fine-grained conjectures. For example, All-Edges Sparse Triangle and Sparse Monochromatic Triangle where each node has degree $n^{\delta}$ for some $\delta < 1$, and the Exact Triangle where edges have arbitrary weights, all have this property under the 3SUM Conjecture. However, as there are slightly nontrivial algorithms for 3SUM, it is natural to wonder if the trivial algorithm for these tricky triangle tasks might also be improved.

Applying a variety of techniques from randomized algorithms, circuit complexity, and communication complexity, we present the first improvements over the trivial algorithms for each of these problems in the Word RAM model. Moreover, our algorithms can be implemented with \emph{only} polysize $\AC0$ operations on words. Extending our techniques, we also show how to solve the notorious 4-cycle detection problem on $n$-node graphs in $o(n^2)$ time, in a Word-RAM model with word size $w > \omega(\log^2 n)$. Along the way, we show how to sort $n$ items over a universe of size $2^u$ using only $\ac0$ word operations in $O(n u \log n)/w$ time.
\end{abstract}

\iffullversion{\bigskip
\bf \emph{Note: This is the full version of a paper that will appear in MFCS 2026.}}
\else{}
\fi

\newpage
\section{Introduction} \label{sec:intro}

Arguably the most fundamental graph problems in fine-grained complexity are concerned with finding small cliques of various kinds, especially triangles (3-cliques). Among the large class of problems that are known to be \emph{subcubic equivalent} to the All-Pairs Shortest Paths (APSP) problem, the most basic such problem is {\sc Negative Triangle}: find a triangle with negative edge sum in an edge-weighted graph \cite{VW18}. Both the quadratic-time hardness of the infamous 3SUM problem and the cubic-time hardness of the APSP problem imply several hardness hypotheses for triangle problems: 
\begin{itemize}
\item the hypothesis that the {\sc Exact Triangle} problem, of finding a triangle with edge sum zero \cite{VW13} in an edge-weighted graph, requires $n^{3-o(1)}$ time, and
\item the hypothesis that the {\sc All-Edges Sparse Triangle} problem, of determining whether each edge in an $n$-node graph with max degree at most $n^{\delta}$ participates in a triangle, requires $n^{1+2\delta-o(1)}$ time~\cite{VX20}. 
\end{itemize}
Similarly, for the {\sc All-Edges Sparse Monochromatic Triangle} problem, where we have an $n$-node graph with colors on its edges and maximum degree at most $n^{\delta}$ and we want to determine which edges participate in monochromatic triangles, it is known that the problem requires $n^{1+2\delta-o(1)}$ time under both the APSP and the 3SUM hypotheses~\cite{LPV20, VX20}.

Already from the work of Fredman in 1976 \cite{Fre76} leading up to the work of Williams \cite{Wil18} on APSP, it has been known that {\sc Negative Triangle} can be solved in $o(n^3)$ time on $n$-node graphs, with the best known running time being $n^3/2^{\Omega(\sqrt{\log n})}$ \cite{Wil18}. In contrast, there have been no reported improvements over the obvious running times for {\sc All-Edges Sparse Triangle}, {\sc Exact Triangle}, and {\sc All-Edges Sparse Monochromatic Triangle}. In this paper, we present improved algorithms for all three of these problems. 

\subparagraph*{Our Results for Triangles in Sparse Graphs.} For a parameter $\delta \in (0,1/2]$, the {\sc All-Edges Sparse Triangle} problem asks: \emph{given an $n$-node graph where each node has degree at most $n^{\delta}$, decide whether each edge of the graph is in a triangle.} 

Although for dense graphs, it is well-known that matrix multiplication leads to faster algorithms~\cite{AYZ97}, the best known algorithm for {\sc All-Edges Sparse Triangle} simply enumerates all pairs of neighbors from each node, taking $O(n^{1+2\delta})$ time. Under standard fine-grained hardness assumptions, there is a matching lower bound. The work of P{\u{a}}tra{\c{s}}cu~\cite{Pat10} showed that the problem requires $n^{1+2\delta - o(1)}$ time under the 3SUM hypothesis.  Vassilevska~W. and Xu~\cite{VX20} showed the same lower bound holds under the Exact Triangle Hypothesis (the hypothesis being that the problem requires $n^{3-o(1)}$ time). We note that this reduction incurs an $O(\log n)$ factor, and therefore shaving a log factor from All-Edges Sparse Triangle does not immediately carry over to Exact Triangle. Chan, Vassilevska~W. and Xu~\cite{DBLP:conf/stoc/ChanWX22} gave a similar conditional lower bound assuming that either APSP or 3SUM over the reals cannot be solved in $n^{3-\eps}$ time or $n^{2-\eps}$ time, respectively.

We present the first improvement over the trivial algorithm for the problem. To keep ourselves honest (disallowing arbitrary multiplications and disallowing arbitrary lookup tables of size $2^w$), our computational model will be the Word RAM model with $\AC0$ operations on words of $w$ bits. That is, our unit-time operations are restricted to those that computed on a constant number of $w$-bit words, by a $\poly(w)$-size $\AC0$ circuit. (Note $w = \Theta(\log n)$ is the standard choice.) This is a standard model; see \Cref{sec:prelims} for more details. 

Please observe we are making a \emph{restriction} on the standard Word RAM model. When $w \leq \delta \log n$ for $\delta < 1$, our algorithms can be implemented with lookup tables for practically any reasonable random-access model with constant-time lookup. Our theorems imply not only that all these problems can be solved faster in the standard Word RAM, \emph{but also} in a model with a more restricted set of operations where no $2^w$-bit lookup tables are allowed (which would permit \emph{arbitrary} operations on $w$-bit words).

\begin{theorem} \label{thm:all-edges} For all $\delta \in (0,1/2]$, there is a Las Vegas algorithm for {\sc All-Edges Sparse Triangle} on graphs with maximum degree $n^{\delta}$ running in $O(n^{1+2\delta} \log w)/ w$ expected time, where $w$ is the word size.
\end{theorem}

Our ideas are general enough that they extend in a natural way to solve the apparently harder \emph{monochromatic} triangle problem in sparse graphs as well. 

\begin{theorem}\label{thm:monochromatic} For all $\delta \in (0,1/2]$, there is a Las Vegas algorithm for {\sc All-Edges Sparse Monochromatic Triangle}~ running in $O(n^{1+2\delta} \log w)/ w$ expected time.
\end{theorem}

Our techniques imply a new data structure for the {\sc Set Intersection} query problem: given sets $S_1,\ldots,S_n$ each of size $O(d)$, we wish to preprocess them so that given a query $(i,j) \in [n] \times [n]$, we can determine if $S_i \cap S_j = \varnothing$ as fast as possible. There has been a line of work on this fundamental problem. Among other results, Bille, Pagh, and Pagh~\cite{BPP07} show how to compute the intersection of $m$ sets with $n$ total elements in time $O(n (\log w)^2/w + k m)$ time, where $k$ is the size of the intersection. Kopelowitz, Pettie, and Porat~\cite{KPP15} give a dynamic data structure (supporting insertions and deletions in $O(1)$ time). When all sets have cardinality at most $d$, their data structure can answer set intersection queries in $O(d \log^2 w)/w$ time. They use two levels of pairwise independent hashing to pack elements into words, using a merging routine of Albers and Hagerup~\cite{AH97} to compute intersections. Eppstein, Goodrich, Mitzenmacher, and Torres~\cite{EGMT17} showed how to remove a $\log w$ factor from these results. In particular, they presented a data structure using cuckoo hash tables and cuckoo filters to determine the intersection of two sets of total size $d$ in $O(d \log w)/w$ expected time. The behavior of their data structure is rather complex, requiring a sophisticated analysis of the ``cyclomatic number'' of random 3-uniform hypergraphs. They also assume \emph{completely random} hash functions (which are necessary in order to prove their tail bounds, see \cite{EGMT17} p.249, bullet 1) which are effectively random lookup tables. In contrast, our algorithms use only simple (linear and $O(\log n)$-wise independent) hashing and $\ac0$ word operations, requiring very little randomness (at most polylogarithmic in $n$). To summarize:

\begin{corollary} There is a data structure for the {\sc Set Intersection} query problem with $\tilde{O}(n k w)$ preprocessing time, such that all queries can be answered in $O(1 + (k \log w)/w)$ time, using only $\poly(\log n)$ randomness in preprocessing and $\ac0$ operations on $w$-bit words.
\end{corollary}

As in \cite{KPP15,EGMT17}, other algorithmic applications follow such as improved triangle listing algorithms; see these papers for references.

\subparagraph*{Our Results for Exact Triangle.} The {\sc Exact Triangle} problem has been widely studied as the basis for a variety of conditional lower bounds (for instance, \cite{VW13, GP18, VW18, VX20, JX23, CX24}). In particular, the hypothesis that the problem requires $n^{3-o(1)}$ time on $n$-node graphs is widely believed, as a faster algorithm would simultaneously refute multiple hardness conjectures in fine-grained complexity. We show how to shave a $\log^{3/2} n$ factor from the $n^3$ time bound in standard models, and a $w^{3/2}$ factor for the $w$-bit Word RAM.

\begin{theorem}\label{thm:exact-tri}
    There is a Las Vegas algorithm for {\sc Exact Triangle} on $n$-node graphs in $O(n^3 (\log w)^{3/2}/ w^{3/2})$ expected time. 
\end{theorem}

\subparagraph*{An Attack on the 4-Cycle Problem.} Finally, we consider the notorious problem of detecting a 4-cycle in an undirected graph on $n$ nodes and $m$ edges. An $O(n^2)$-time algorithm is known~\cite{richards1985finding} and an $O(m^{4/3})$-time algorithm is known~\cite{AYZ97}, but these running times coincide when the number of edges is $\Theta(n^{1.5})$. Indeed, the ``hard case'' of 4-cycle detection is when the graph has $O(n^{1.5})$ edges: Bondy and Simonovits showed that any graph with $cn^{1.5}$ edges for sufficiently large constant $c$ must have a 4-cycle \cite{BS74}, and works due to Brown and Erd\H{o}s, R\'enyi, and S\'os gave algebraic constructions of 4-cycle free graphs that match this bound \cite{Bro66, ERS66}. There are essentially no hardness results known for the detection version of this problem, and yet no algorithmic improvements are known beyond the $O(n^2)$-time solution, which enumerates pairs of neighbors of nodes and looks for a collision. We note that by contrast, much progress has been made on 4-cycle listing. In particular, there is a $O(\min(n^2, m^{4/3}) + t)$ time algorithm \cite{AKLS22}, where $t$ is the number of 4-cycles, along with a matching lower bound under the 3SUM hypothesis \cite{ABF22, JX23}.

We give two results on 4-cycle detection. First, building on our ideas for finding triangles in sparse graphs, we reduce 4-cycle detection to a variant of the Element Distinctness problem.

\begin{definition}[Concatenated Element Distinctness]
We are given $k$ lists $L_1, \dots, L_k$, where each list has $n$ elements of $\ell$-bit strings. For every $i, j \in [k]$, element-wise concatenate them to a produce a new list $L_{i, j}$ containing $n$ elements, each of which now have $2\ell$ bits. The task is to decide whether or not for all $i \neq j$, all $n$ strings in $L_{i,j}$ are distinct.  
\end{definition}

In other words, the $k$ lists form $O(k^2)$ instances of Element Distinctness on $n$ elements, and we accept if and only if all of the generated instances have distinct elements.

\begin{theorem}\label{thm:reduction} For $T(n,k,\ell) \geq n+k+\ell$, if Concatenated Element Distinctness can be solved in $O(T(n,k,\ell))$ time, then 4-cycle detection in $n$-node graphs can be done in $O(T(n,\sqrt{n},(\log_2 n)/2))$ time.
\end{theorem}

Motivated by this reduction, we show how to sort $n$ items over a universe of size $2^u$ in $O(n u \log n)/w$ time in the $\ac0$ Word RAM model on $w$-bit words. Prior work either required more complex word operations such as $2^w$-size lookup tables \cite{Cha10}, or extra logarithmic factors 
\cite{AH97}. As sorting easily solves Element Distinctness, we obtain the following:

\begin{theorem}\label{thm:c4det} We can detect a 4-cycle in $n$-node graphs in $O(n^2 \log^2 n)/w$ time in the $\AC0$ Word RAM.
\end{theorem}

For word size $w \leq \log^2 n$, this result is suboptimal, as we already know an $O(n^2)$-time algorithm when $w=\Theta(\log n)$. Nevertheless, it was not obvious how to improve over $O(n^2)$ time even for large word sizes: the known algorithms involve making $O(n^2)$ unstructured probes into a table, enumerating all pairs of neighbors of a given node to find a common pair of neighbors. We believe the reduction of \Cref{thm:reduction} may be useful for future work.

\subparagraph*{Sublinear Time 4-Cycle Detection in Denser Graphs.} Finally, we observe that 4-cycles can be found surprisingly quickly in graphs that have slightly more than $n^{1.5}$ edges. In fact, they can be found in \emph{sublinear time} in the number of edges. 

\begin{theorem}\label{thm:sublinear}
There is a Las Vegas algorithm that, given a graph with $n$ vertices and at least $n^{3/2 + \delta}$ edges for $\delta \in (0, \frac{1}{2})$, returns a 4-cycle and runs in expected $O(n^{5/4-\delta/2})$ time. 
\end{theorem}

Therefore, on graphs with $\eps n^{3/2}$ edges where $\eps > 0$ is a small constant, the best known method for detecting (and finding) a 4-cycle remains $O(n^2)$ time. However, on graphs with $\eps n^{3/2 + \delta}$ edges for any positive $\delta$, the time complexity of 4-cycle surprisingly drops to $O(n^{5/4-\delta/2})$ time. 

A classic supersaturation result due to Erd\H{o}s and Simonovits \cite{ES83} states that a graph with $m \gg n^{3/2}$ edges has at least $\Omega((m/n)^4))$ 4-cycles, so there is a natural \emph{random sampling} algorithm for finding 4-cycles: simply try uniform random 4-tuples of nodes! Our algorithm outperforms the random sampling algorithm when the number of edges is not significantly larger than $n^{3/2}$. For instance, when $m = \Theta(n^{3/2+\delta})$, supersaturation guarantees $\Omega(n^{2+4\delta})$ 4-cycles, so random sampling takes expected $O(n^{2-4\delta})$ time. This is strictly slower than our method when $\delta \in \left(0, \frac{3}{14}\right)$. Combining the two approaches, we obtain:

\begin{corollary}
There is a Las Vegas algorithm that, given a graph with at least $n^{3/2 + \delta}$ edges for $\delta \in (0, \frac{1}{2})$, returns a 4-cycle and runs in expected $O(n^{\min\{5/4-\delta/2,2-4\delta\}})$ time. 
\end{corollary}

\iffullversion{}
\else {\bf \emph{Note: some proofs are omitted due to space constraints.}}
\fi

\subsection{Intuition}

Let us give some intuition for the new ideas and techniques that go into our results, starting with our results for triangles in sparse graphs. Fundamentally, triangle finding in sparse graphs amounts to computing a series of intersections of sparse sets. While there is a considerable literature on data structures for set intersection (for example \cite{BPP07, DBLP:journals/tcs/CohenP10, DBLP:journals/pvldb/DingK11, KPP15, EGMT17}), prior work has focused on complex hashing schemes with complex word operations. 

Our approach reduces triangle problems in sparse graphs to problems which are more easily ``word-packable'', and can be cleanly implemented in simple ($\ac0$) word operations. For example, in \Cref{thm:warmup} we show how to randomly embed set intersections of sets of size $k$ into the \emph{equality inner product problem}, in which we are given vectors $u,v$ of length $O(k)$ with entries in $[n]$, and we wish to find a component $i$ such that $u_i = v_i$. Since equality inner products are easy to pack into word operations, this already suffices to detect triangle in sparse graphs. To solve the all-edges case, we partition the universe of $n$ nodes into $O(n^{\delta}/\log n)$ components using $O(\log n)$-wise independent hashing, and argue that each vertex neighborhood has $O(\log n)$ nodes in each component. We can then compute the $O(\log n)$-set intersections quickly by adapting a randomized communication protocol for set intersection~\cite{DKS12}. We use similar ideas to get our reduction from 4-cycle to {\sc Concatenated Element Distinctness}.

For the {\sc Exact Triangle} problem, we can do similar tricks. However we can obtain an improved speedup (from a $\log n$ factor to a $\log^{3/2} n$ factor) by exploiting the density of the graph. In particular, we can pack small $O(\sqrt{\log n})$-size subgraphs of the graph in our words, along with linear hash values of weights.

In the process of proving our results on 4-cycle, we present an improved ``packed sorting'' algorithm. In this problem, we wish to sort a list of $n$ items from the universe $[\poly(n)]$ where the unsorted list is provided in $(n \log n)/w$ words, each word containing $O(w/\log n)$ items. Although it was known how to sort a packed list in $O(n \log^2 n)/w$ time using arbitrary word operations, we show how to do it using only $\ac0$ word operations. The key is that, once the sublists within each word are sorted themselves, all that remains is to merge the sublists. We show how to perform merges in $\AC0$, by exploiting the sorted order to determine in parallel the ranks of all items in the merge.

Although our results only give ``minor'' running time improvements (polylog-shaving and log-log shaving) over the prior best-known algorithms, we achieve these results with \emph{simple} techniques, by viewing these fundamental and well-studied problems in a different light. It is reasonable to believe that applying more sophisticated methods could lead to stronger improvements.

\section{Preliminaries} \label{sec:prelims}

\subparagraph*{Model of Computation.} We use the $\ac0$ Word RAM model of computation with word size $w$. This means that we \emph{restrict} ourselves to word operations that can be computed by a $\poly(w)$-size circuit of constant depth with AND and OR gates of unbounded fan-in as well as NOT gates, and assume that any such operation can be computed in constant time~\cite{H98, AMT99, BDP08}. Most standard Word RAM instructions can be implemented in the $\ac0$ Word RAM, including addition and subtraction modulo $2^w$, bitwise boolean operations, comparisons, and shifts \cite{SV84, CSV84}. The only standard operation not supported in this model is multiplication. For wordsize $w \gg \log n$, we explicitly disallow $2^w$-size lookup tables, which would allow us to compute \emph{any} function on $w$ bits: we only allow $\ac0$ functions to be computed on $w$-bit words in constant time.

\subparagraph*{Graphs.} Unless otherwise stated, all graphs are undirected and unweighted.

\begin{theorem}[Reiman] \label{thm:reiman}
    If $G$ is a four-cycle free graph on $n$ nodes, it has at most $\frac{n}{4}\left(1 + \sqrt{4n - 3}\right)$ edges.
\end{theorem}

\subparagraph*{Chernoff lower bounds for $k$-wise independence.} We will also use the following adaptation of Chernoff bounds to $k$-wise independent random variables:

\begin{theorem}[\cite{DBLP:journals/siamdm/SchmidtSS95}] \label{thm:chernoff}
Let $k$ be a positive integer, and let $X$ be the sum of $0/1$ valued $k$-wise independent random variables, each with expectation $\mu$. For $\delta \geq 1$ and $k \leq \lfloor \delta \mu/e^{1/3}\rfloor$, we have $\Pr[X \geq (1+\delta)\mu] \leq e^{-k/2}$.
\end{theorem}

\section{Compressed List Operations in the AC0-RAM} \label{sec:wordtricks}
At a high level, our algorithms use various hashing schemes to reduce triangle and 4-cycle to problems on compressed lists, in particular, equality product, set disjointness, and element distinctness. We then take advantage of word-level parallelism and known techniques for designing $\AC0$ circuits to solve these problems faster, giving us word size factor savings. We begin by detailing these compressed list operations.

\begin{definition}[Compressed list representation]
    Let $L$ be a list of $k$ elements from $[2^{\ell}]$, where $\ell \leq w$. We store such $L$ in \emph{compressed list representation} by packing $\lfloor w/\ell \rfloor$ contiguous list elements into a single word, for an overall packed length of $O(k\ell/w)$ words. 
\end{definition}

In the following, we let $\bot$ denote a special character that does not match any valid element. This can be easily implemented in various ways; e.g., we can extend the $\ell$-bit encoding of every string to $\ell+1$ bits, where all ``valid'' $\ell$-bit elements start with a zero, and $\bot$ equals all-ones (for example).

\subsection{Sorting Compressed Lists}

We begin by describing a sorting algorithm on compressed lists with $\AC0$ word operations, loosely inspired by a work of Chan \cite{Cha10} who allowed \emph{arbitrary} word operations. Given two compressed sorted lists $L_0$ and $L_1$, we define the \emph{merge step of $L_0$ and $L_1$} to be the task of implementing the standard linear-time merge of $L_0$ and $L_1$ until one of the lists $L_b$ is empty. The output of a merge step is the merged list obtained so far, along with a $b\in\{0,1\}$ and the remaining content of the non-empty list $L_b$. We give a procedure to perform this operation on single words in compressed list representation, which we will use as a subroutine in the overall merge. In particular, at every step of the high-level merge, we will run the single-word merge on the leftover word from $L_b$ with the next full word from $L_{\overline{b}}.$ 

\begin{lemma} \label{claim:const_merge} 
    Let $w$ be the word size. Given compressed sorted lists $L_0$ and $L_1$ where both lists are stored in a single word, the merge step of $L_0$ and $L_1$ can be implemented in constant time with $\ac0$ word operations. We assume list elements are in the universe $[2^u]$ for $u \leq w$. 
\end{lemma}

\begin{proof}[Proof of \Cref{claim:const_merge}]
     Since $L_0$ and $L_1$ are stored in a single word, we may assume there are at most $k = \lfloor w/u \rfloor$ elements in each. If the lists have fewer than $k$ elements, we expect the entries to appear packed consecutively in the upper bits of the word, and the remainder of the word to be padded with $u$-bit $\bot$ entries which are treated as ``infinity'', larger than all other entries. 
    
    To sort $L_0 \cup L_1$, observe that it suffices to determine, for every element $x$ in $L_0$, the number of elements in $L_1$ that precede $x$ in sorted order, and for every element $x'$ in $L_1$, the number of elements in $L_0$ that precede $x'$ in sorted order. Using an $\AC0$ circuit, exploiting the fact that $L_0$ and $L_1$ are already sorted, we can check all possible values for these numbers in parallel. We will perform comparisons as if all entries in $L_0$ have a trailing 0 appended to them, and all entries in $L_1$ have a trailing 1. This means that there does not exist any element that appears in both $L_0$ and $L_1$, which has the effect of breaking ties in favor of $L_0$. 

    To compute the $k$th value of the merged output, we check the following. For every $i \in [|L_0|]$, compute $j = k - i.$ Then, check the following conditions: 
    \begin{enumerate}
        \item[(1)] $L_0[i] > L_1[j]$ and
        \item[(2)] either $L_0[i] < L_1[j+1]$ or $L_1[j] = \bot$.
    \end{enumerate}
    If the AND of (1) and (2) holds, we output $L_0[i],$ else output 0. We then OR together all of these gadgets for all $i \in [|L_0|],$ and also OR these results with the analogous computation taking $i \in [|L_1|].$ If the condition is satisfied (say for $i \in [|L_0|]$), this means that there are $j+1$ elements in $L_1$ that are strictly smaller than $L_0[i],$ and the rest of the elements are strictly larger. Since elements in $L_0$ cannot be equal to those in $L_1$, there is only one such value of $j$ per element. Additionally, since we enforce that our merge step is stable, the (unique) correct location in the merged list for this element is $i+j,$ as desired. 
    
    As mentioned in \Cref{sec:prelims}, the addition and comparisons can be done in $\AC0$. By construction, at most one comparison gadget fires a nonzero index, so we can simply OR them all. Thus, we can run these checks for all $|L_0| + |L_1|$ list elements in parallel, constructing our merged list index by index. 

    All that remains is to determine which list will be left empty following the merge step. To do this, let $i$ and $j$ be the last elements of $L_0$ and $L_1$. If $i < j$, then $L_0$ is the newly emptied list ($b=0$), otherwise $L_1$ is emptied ($b=1$). Let $m$ be the minimum of $i$ and $j$. 
    We place all elements whose sorted rank is in $[m]$ into at most two words, which is the merged portion. We can pad the second word with $\bot$. We place the remaining non-$\bot$ elements (all from list $L_{\overline{b}}$) into another word (again padding with $\bot$), and return $b$, completing the merge step.

    As mentioned above, each of these gadgets is constant-depth, and we compose them in parallel. For every element, we perform at most $\max(|L_0|, |L_1|)$ comparisons, so the circuit size is at most quadratic in the input size, and thus the entire construction is in $\ac0.$   
\end{proof}

\begin{claim} \label{claim:lin_merge} Given compressed sorted lists $L_0$ and $L_1$ with $u$-bit entries and at most $k$ elements each, a merge of $L_0$ and $L_1$ can be implemented in $O(k u/w)$ time with $\ac0$ word operations.
\end{claim}

\iffullversion
\begin{proof}[Proof of \Cref{claim:lin_merge}]
    We will use \Cref{claim:const_merge} as a subroutine. We know that $L_0$ and $L_1$ are word-packed, so we can merge them one word at a time. In particular, start by merging the first words of each list. Then, take the remainder word produced by the merge step, along with the next full word of the list that ran out, and merge those. Repeat that process until one of $L_0$ and $L_1$ is completely exhausted, at which point we concatenate the final remainder word with the merged list we have built up so far. Each of the merge steps takes constant time, since it is implementable via an $\ac0$ circuit as is described in  \Cref{claim:const_merge}. At every iteration, we process a new word from either $L_0$ or $L_1.$ Thus, this algorithm terminates in $2ku/w$ steps, which takes $O(ku/w)$ time. 

\end{proof}
\fi

Now that we have a merge step, it is natural to construct a mergesort-like sorting algorithm. We note that this gives an improvement over the best-known packed sorting algorithm in the Word-RAM model, which is due to Belazzougui, Brodal, and Nielsen \cite{BBN14}. Their algorithm sorts $n$ integers on $w/u$ bits packed into compressed lists in $O(\frac{n}{u}(\log n + \log^2 u))$ time, and works for any choice of $w = \Omega(\log n)$. In particular, they are interested in the setting where each element is on $u = O(\log n)$ bits, so $u = O(w/\log n).$ Thus, their overall runtime is $O(\frac{n\log n}{w}(\log n + \log^2 w)),$ which we will slightly improve.

\begin{claim} \label{claim:log2sort}
    Given a list $L$ of $n$ elements in $[2^u]$ packed into words, we can sort $L$ in $O(n u \log n / w)$ time.
\end{claim}

\begin{proof}[Proof of \Cref{claim:log2sort}]
    We are able to compute the result of one $\ac0$ circuit that takes $w$ bits as input in constant time. Thus, in $\log w$ time, we have access to circuits with unbounded fan-in and logarithmic depth, which captures at least all of $\nc1,$ which, crucially, contains sorting~\cite{AjtaiKS83}. So we require $O(\log w)$ time to sort the elements in a single word.

    This serves as the preprocessing for our merging procedure, and runs in $O(n\log w / w)$ time. The bulk of the computation is the merge step. We can now merge each of the $n\log n/w$ words in a tree-like fashion. One can think of placing each word at the leaf of a height-$O(\log n)$ binary tree, and then storing the merged version of the children in every parent node. The recurrence associated with this is $T(m) = 2T(m/2) + O(m),$ so we require $O(m\log m)$ time, where $m := nu/w$ is the total number of words. Thus the overall runtime is $O(nu \log n/w)$, which dominates the preprocessing.
\end{proof}

We note that the difficulty of packed sorting may be inherent. In particular, if one could remove the $\log n$-factor from the runtime above, this would resolve the longstanding open problem of linear-time sorting for (almost) all word sizes. The Belazzougui, Brodal, and Nielsen paper achieves linear-time for $w = \Omega(\log^{2}n \log\log n),$ leaving the problem open only for word sizes that are $\omega(\log n)$ (from radix sort) and $o(\log^2 n \log\log n).$ However, the paper does this by reducing sorting on word size $w$ to $\log\log n$ levels of packed sorting of $O(n)$ elements on $O(\log n)$ bits. Thus, shaving off the log factor on the packed sorting runtime would give a $O(n \log n \log\log n / w)$ time algorithm, which is linear when $w = \Omega(\log n \log \log n),$ almost completely closing this longstanding gap in word size.

It is also known that shaving this log-factor, or more generally, solving packed sorting in time faster than $o(\min\{n, \frac{nu}{w} \log \frac{nu}{w}\})$, would refute the Undirected $k$-Pairs Conjecture, which is central to the field of network coding \cite{LL04, FHLS19}.

\subsection{Indicator Masks for Compressed Lists}

Beyond sorting compressed lists, we would like to efficiently indicate list elements of interest to our algorithms. To do this, we define an \emph{indicator mask} to be a length-$w$ bitstring that is all zeros except the leading bit of each entry. This bit will indicate the value of some predicate evaluated on the associated entry in another word. We now show, inspired by bit tricks from \cite{fredman-willard}, that we can produce a mask that indicates the all-zeros entries in $O(1)$ time.

\begin{claim} \label{claim:const-all-zeros}
    Let $z$ be a word in a compressed list. There is a constant-time algorithm to produce an indicator mask whose ones indicate that the associated entries in $z$ are all-zeros.
\end{claim}

\iffullversion
\begin{proof}[Proof of \Cref{claim:const-all-zeros}]
    The high-level idea is as follows: we will produce a mask $I_L$ that indicates whether the leading entry is a zero, a second mask $I_R$ that indicates whether the remaining $\ell-1$ are all-zeros, and return $I_L \land I_R$.

    The first mask is easy to construct. Let $I$ be the indicator mask whose indicator bits are all ones. Let $I_L := I \land (\lnot z).$ This sets all indicator bits to $1$ if and only if the leading bits in their associated entries were 0. We then use a bitmask to zero out non-leading bits in $I_L$. 

    Constructing $I_R$ requires a bit more work. First, we need to zero out the leading bits of every entry while leaving the remaining bits untouched. To do this, we compute $R := z \land (\lnot I).$ Now, consider $D := I - R.$ For each entry, since the leading bit of $I$ is a 1 while that of $R$ is a 0, there is no borrowing across entries. If the value in $R$ is nonzero, then the subtraction requires borrowing from the leading 1 in the corresponding entry in $I$. Thus the entry has the form $0*$. If the value in $R$ is all-zeros, the entry is $100\dots$ as before. Thus, $I_R := D \land I$ is an indicator mask for the entries whose non-leading bits were all-zeros in $z$, as desired. We return $I_L \land I_R.$ Because the computation required a constant number of word operations, we get a runtime of $O(1)$.
\end{proof}
\fi

Our next task is to compute an indicator mask for set intersection.

\begin{claim} \label{claim:set-intersect}
    Given two sorted compressed lists $L_0$ and $L_1$ with $k$ elements of length $\ell$ bits, which correspond to $k$-element sets $S_0$ and $S_1$, we can compute a mask that indicates elements in $S_0 \cap S_1$ in $O(k\ell /w)$ time.
\end{claim}

\iffullversion
\begin{proof}[Proof of \Cref{claim:set-intersect}]
    Since $L_0$ and $L_1$ correspond to sets, we assume they do not contain duplicates. We begin by merging $L_0$ and $L_1$ to produce a new compressed list $L$. By \Cref{claim:lin_merge}, this can be done in $O(k\ell / w)$ time. We then make a copy of $L$ and shift it to the left $\ell$ bits. We XOR $L$ with the shifted copy. This causes any adjacent duplicates (which correspond to items in $S_0 \cap S_1$) to have all-zeroes entries of length $\ell$ in locations corresponding to indices in $L_0$. We can then use \Cref{claim:const-all-zeros} to produce an indicator mask with ones in these locations in constant time. The runtime is thus dominated by the merge step.
\end{proof}
\fi

The purpose of computing these indicator masks is to efficiently indicate ``special'' list entries that are of interest to our algorithms. However, these compressed lists are created as a result of hashing, and our algorithms generally need to filter out false positives by looking up the associated entries in the preimages. To that end, we need to be able to extract indices of the ones in our masks. We will use the following $\ac0$ circuit to do this:

\begin{claim} \label{claim:msb}
    In the $\ac0$-RAM model, there is a constant-time algorithm that returns the index of the most significant set bit of a word. 
\end{claim}

\iffullversion
\begin{proof}[Proof of \Cref{claim:msb}]
    Our goal is to produce a length-$\log w$ index indicating the most significant set bit. There are $w$ possibilities, and we can construct a depth-one circuit that checks each one. In particular, to check if the $i$th bit is the leftmost set bit, we AND $i$ with the negation of each of the $i-1$ bits to its left. 
    
    We now need to reconstruct the index. To do this, we add an additional layer of OR gates $G_j$ for every $j \in [\log w]$. We would like for the concatenation of all the $G_j$'s in order to produce the index. Fix one of the circuits above, and let $i$ be the input bit that it is associated with. We will connect its output to every $G_j$ such that the $j$th bit of the binary representation of $i$ is one. Correctness follows from the fact that at most one of the depth-two circuits will output a 1, so only the $G_j$'s associated with the set bits of that index will be 1. The bottom level of the circuit has $w$ input bits, the fan-in of each of the $n$ ANDs is at most $w$, and the fan-in of each of the $\log w$ ORs is at most $n/2$. Thus, the circuit has polynomial size and depth two, and is therefore $\ac0$. 
\end{proof}
\fi

Note that we can call this procedure, zero out the associated bit, and then repeat to extract the indices of all of the ones in a particular indicator mask. This only requires us to pay $O(1)$ time per set bit whose index we retrieve. 

\section{A Faster Algorithm for All-Edges Sparse Triangle}

Armed with the compressed list operations from above, we are ready to solve triangle problems faster. We begin with a subquadratic algorithm for Sparse Triangle Detection, which is quite natural but is not strong enough to solve the All-Edges version faster. The goal here is to emphasize the relationship between triangle problems and equality product-like problems, and to motivate the reduction to set intersection that we will present in \Cref{sec:AESparseTri}.

\subsection{Warm-Up: Solving Sparse Triangle Detection Faster} \label{sec:warmup}

We will solve sparse triangle detection by reducing to the following problem:

\begin{definition}[All-Edges Equality Inner Product]
    Let $\delta \in (0, 0.5]$ and $d = O(n^\delta)$. Given an $n \times d$ matrix $A$ and an $d \times n$ matrix $B$ with entries in $[n^{1-\delta}] \cup \{\bot\}$, as well as a set $P$ of $O(n^{1+\delta})$ points in $[n] \times [n]$, is there an $(i, j)$ in $P$ and a value of $k \in [n^\delta]$ such that $v := A[i, k] = B[k, j]$ and $v \neq \bot$? 
\end{definition}

\begin{theorem}\label{thm:warmup}
    There is a $O(n^{1+\delta})$-time randomized reduction from the triangle detection problem on graphs of maximum degree $O(n^\delta)$ to All-Edges Equality Inner Product. When there is no triangle in the original graph, the All-Edges Equality Inner Product instance is a no instance. When there is a triangle, the All-Edges Equality Inner Product instance is a yes instance with constant nonzero probability.
\end{theorem}

\begin{proof}[Proof of \Cref{thm:warmup}]
We assume without loss of generality that $G$ is tripartite, with parts $I$, $J$, and $K$ each of size at most $n$. We begin by ``color-coding'' vertices as in \cite{AYZ95}, but with a large color palette. We impose the following rule: every vertex in $I \cup J$ has an edge to at most one vertex per color class of $K$. To enforce this, if a vertex $u \in (I \cup J)$ has multiple edges to one color class, we will select one uniformly at random to keep and delete the rest. We show that every triangle $(u,v, w)$ in the graph survives with constant probability.
    
    Suppose we have a pairwise independent hash family $\calH
    = \{h: [n] \to [C n^\delta] \}$ for large constant $C$, where each $h \in \calH$ corresponds to a coloring of the vertices of $K$ with $C n^{\delta}$ colors. 
    Suppose $u \in I,$ $v \in J$ and $w \in K$ form a triangle. Let $N(u)$ and $N(v)$ denote the neighbors of $u$ and $v$ in $K$, respectively. The probability that the triangle $u, v, w$ survives is the probability that both $(u, w)$ and $(v, w)$ survive. Fix a color $c \in [C n^{\delta}]$. We have:
    \begin{align*}
        \Pr[&(u, v, w) \text{ survives and $h(w) = c$}]\\ &=\Pr_{h \gets \calH}\left[h(w) = c \land (\forall k \in (N(u) \cup N(v)) \setminus \{w\}) ~h(k) \neq c \right] \\
        &= \frac{1}{Cn^\delta} \Pr_{h \gets \calH}\left[(\forall k \in (N(u) \cup N(v)) \setminus \{w\}) ~h(k) \neq c \mid h(w) = c \right] \\
        &\geq \frac{1}{Cn^\delta} \left(1 - \Pr_{h \gets \calH}\left[(\exists k \in (N(u) \cup N(v)) \setminus \{w\}) ~h(k) = c \mid h(w) = c \right] \right) \\
        &\geq \frac{1}{Cn^\delta} \left(1 - |(N(u) \cup N(v)) \setminus \{w\}| \cdot \frac{1}{Cn^\delta} \right) \\
        &\geq \frac{1}{Cn^\delta} \left(1 - 2n^\delta \cdot \frac{1}{Cn^\delta} \right) \geq \frac{1}{Cn^\delta}\left(1-\frac{2}{C}\right).
    \end{align*}

    For each of the $C n^\delta$ choices of $c$, all of the events in question are disjoint. Summing over all of them, the probability $(u,v,w)$ survives is at least $1-2/C$. Setting $C \geq 1$ large, we conclude that every triangle in the original graph survives with constant probability. 
    
    We now use the color classes to express the remaining graph as a smaller matrix. We note that with probability at least $1-1/D^2$, all of the color classes have size at most $Dn^{1-\delta}$; in the end, we will choose $D$ to be a large constant. To see this, let $X_c := \sum_{w \in K} \mathbbm{1}[h(w) = k]$, the number of vertices that are assigned color class $c$. Note that $\E[X_c] = n^{1-\delta}/C.$ By pairwise independence, $\mathsf{Var}(X_c) = \sum_{w \in K} \mathsf{Var}(\mathbbm{1}[h(w) = k]) \le n^{1-\delta}/C.$ By Chebyshev's inequality,
    \[\Pr[X_c - \E[X_c] \ge Dn^{1-\delta}] \le \frac{1}{CD^2 n^{1-\delta}}.\]
    Union bounding over all $Cn^\delta$ color classes, we have that all color classes have size at most $D n^{1-\delta}$ with probability at least $1-1/D^2$. We construct an $n \times n^\delta$ matrix $A$ that captures edges between $I$ and $K$, as well as a $n^\delta \times n$ matrix $B$ that captures edges between $K$ and $J$. Note that all indices $i, j$ are in $[n]$ and index their respective vertex sets, whereas the $k$ indices are in $[O(n^{1-\delta})]$ and determine a color class of $K.$ We also assume that we re-index the nodes in each color class of $K$, so that the tuples $(k, v) \in [n^\delta] \times [O(n^{1-\delta})]$ uniquely identify some node in the $k$th color class of $K$. 
    
    We now set the $(i, k)$ entry of $A$ to the node in the $k$th color class of $K$ that the vertex $i$ has an edge to (and the $(k, j)$ entry of $B$ analogously). If the vertex has no edge to this color class, set the entry to $\bot$. Now, if there exists a $k$ such that $A[i, k] = B[k, j] = v$ and $v \neq \bot$, both $i$ and $j$ share an edge to the $v$th node of the $k$th color class in $K$. Thus, if $(i, j)$ is itself an edge, then $i, j, (k, v)$ form a triangle. 
    
    Thus, with constant probability, the All-Edges Equality Inner Product problem on the matrices $A$ and $B$ with $P := E,$ the edge set of $G,$ returns the same answer as the Sparse Triangle Detection problem on the original graph. 
\end{proof}

We now give an algorithm to solve All-Edges Equality Inner Product faster than the obvious $n^{1+2\delta}$ time solution. The high-level idea is as follows: we use a linear hash function to hash the matrix entries down to $O(\log w)$ bits, write the problem down in a compressed matrix form by packing multiple entries into a single word, and then use word operations to compute the equality inner product. 

\begin{theorem}\label{thm:tri-det}
    For every $\delta \in (0,1/2]$ and word size $w \geq \Omega(\log n)$, there is a Las Vegas $O(n^{1+2\delta} (\log w) / w )$-time algorithm for the All-Edges Equality Inner Product problem on $n \times n^\delta$ and $n^\delta \times n$ matrices.
\end{theorem}

\iffullversion
\begin{proof}[Proof of \Cref{thm:tri-det}]
    The entries of each matrix are in $[O(n^{1-\delta})],$ and thus can each be written down in $w$ bits. Our goal is to pack these matrices into a smaller space, which will improve query time. We will create new matrices $A'$ and $B'$ in the following way. We define $h_r(x) := \langle r, x \rangle,$ where the inner product is taken mod 2. For each $k \in [n^\delta]$, we draw $\ell - 2$ bitstrings $r_1, \dots, r_{\ell-2}$ uniformly at random from $\{0, 1\}^{w},$ for $\ell$ to be set later. For every non-null element  $x$ in row $k$ of $A$ (or column $k$ of $B$), replace $x$ with the length-$\ell$ bitstring $h_{r_1}(x) \cdots h_{r_{\ell}}(x)00$ in the corresponding location of either $A'$ or $B'$. If $x$ is $\bot$ and an entry of $A$, put $0\dots001$ in $A'$. If $x$ is $\bot$ and in $B,$ put $0\dots 010$ in $B'$. This way, the hash values of null entries cannot equal one another or anything else.  Note that if $x \neq y$, we have that $\Pr_r[h_r(x) = h_r(y)] = 1/2.$ Thus, over all $\ell-2$ inner products, the probability that $x$ and $y$ hash to the same value is $1/2^{\ell-2}$. 

    Each row of $A'$ (and column of $B'$) has $n^\delta$ entries, each of which is a length-$\ell$ bitstring. Then we require $n^\delta \ell / w$ words to store each one. Our queries are of the form $(i, j)$, so we can think of them as computing equality product on the $n^\delta \ell / w$ registers associated with the $i$th row of $A'$ and the $j$th column of $B'$.
    
    We set $\ell = 3\log w,$ so that the probability that two distinct values hash to the same bitstring (a false positive) is $4/w^3$. To achieve the desired runtime, we need to answer each query in $n^\delta \ell/w$ time, which is precisely the number of pairs of words we are operating on. For each pair of words that make up the relevant row/column pair, XOR them. Note that any matching entries become the all-zeroes entry, and we can create an indicator mask for them using  \Cref{claim:const-all-zeros}.
    Then, we use \Cref{claim:msb} to extract the index of the most significant set bit. This index can be translated to a value of $k \in [n^\delta]$ provided we keep track of which registers we are checking, and then we can look up the original entries $A[i, k]$ and $B[k, j]$ to verify. If it turns out that this value is a false positive, we can set this bit to zero and then rerun the algorithm to find the new most significant set bit, which corresponds to another witness.

    However, we would like to keep our runtime per query at $n^\delta \log w / w.$ Since identifying and looking up a false positive, we can check at most $O(n^\delta \log w)/ w$ of them per query. Fixing a particular query, let $P$ be the random variable corresponding to the number of false positives in the associated compressed equality inner product calculation. By construction, we have $\E[P] = 4 n^\delta/w^3$. By Markov's inequality, \[\Pr[P \geq 100n^\delta \log w / w] \leq \frac{1}{25 w^2 \cdot \log w} \leq \frac{1}{25}.\]
    
    We can keep a counter for the number of false positives we process per query, and if we exceed $100 n^\delta \log w / w,$ we will rerun the query in a subsequent round of the algorithm. The time per edge query is $O(n^\delta \log w / w),$ and we will run the algorithm on $n^{1+\delta}$ edges in the first round, $n^{1+\delta}/25$ in the second, and so on until all queries have definitive answers. Thus, the total number of edge queries required is bounded by \[n^{1+\delta} \cdot \sum_{i=0}^\infty \left(\frac{1}{25}\right)^i \leq \frac{n^{1+\delta}}{1-\frac{1}{25}} = O(n^{1+\delta}).\] Thus, the expected runtime is $O(n^{1+2\delta} \log w / w)$.
\end{proof}
\fi

While this algorithm works well for triangle \emph{detection}, the color coding reduction preserves every triangle with constant probability. Thus, if we tried to use this approach to solve All-Edges Sparse Triangle, we would answer each edge query correctly with constant probability. However, we would like the following stronger correctness guarantee: the entire $O(n^{1+\delta})$-bit output should be correct with high probability. To do this, we need to circumvent the reduction to All-Edges Equality Inner Product.

\subsection[AE-Sparse-Tri]{Solving All-Edges Sparse Triangle in $O(n^{1+2\delta} \log w/ w)$ Time} \label{sec:AESparseTri}

We begin by formally defining the All-Edges Sparse Triangle problem. We note that the instances that arise in the literature from the APSP and 3SUM reductions have $m = O(n^{1.5})$ and maximum degree $\Delta = O(\sqrt{n}).$ We give a slight generalization, allowing for graphs of any sufficiently small sparsity. However, our techniques do require that the maximum degree $\Delta = O(m/n).$ We also assume without loss of generality that the graph is tripartite; if this is not the case, generate two more copies of the vertex set, and add the original edges of the graph between each of the pairs of parts. 

\begin{definition}[All-Edges Sparse Triangle]
    Let $G$ be a tripartite graph whose vertex set is $I \cup J \cup K$ such that $|I| = |J| = |K| = n.$ Let the number of edges be $m = n^{1+\delta}$ for some $\delta \in (0, 0.5],$ and let the maximum degree be $\Delta := O(n^\delta)$. For every edge $e \in E \cap (I \times J)$, decide whether $e$ participates in a triangle. 
\end{definition}

\begin{reminder}{\Cref{thm:all-edges}}
    There is a Las Vegas algorithm for {\sc All-Edges Sparse Triangle} on graphs with maximum degree $n^{\delta}$ running in $O(n^{1+2\delta} \log w)/ w$ expected time.
\end{reminder}

\begin{proof}[Proof of \Cref{thm:all-edges}]
    Let the word size $w \geq  c\log(n)$ for a constant $c > 0$. As in \Cref{sec:warmup}, we can think of solving All-Edges Sparse Triangle as responding to $O(n^{1+\delta})$ edge queries, each in $o(n^\delta)$ time. Our algorithm is loosely inspired by a communication complexity protocol for set disjointness due to Dasgupta, Kumar, and Sivakumar \cite{DKS12}. We do the following: 
    
    \begin{enumerate}
        \item Randomly partition $K$ into $n^\delta/w$ color classes.
        \item Compress the name of each vertex in $K$ using a random hash function: $h: [n] \to [\poly(\ell)]$.
        \item For each vertex $i$ in $I \cup J,$ for every color class $v$ in $[n^\delta / w]$, write down a compressed list that stores the hashes of their neighbors in that color class in sorted order.
        \item To respond to an edge query $(i, j)$, for every color class, compute the set intersection of the compressed lists associated with $i$ and $j$'s neighbors with that color. If it is non-empty, $(i, j)$ participates in a triangle.
    \end{enumerate}

    We will now discuss each step in more detail. We begin by color-coding the vertex set, randomly partitioning $[n]$ into $n^\delta/w$ color classes $S_1, \dots, S_{n^\delta/w}$. We need to draw $h$ from an $O(\log n)$-wise independent hash family, as follows. 
    
    Fix an arbitrary vertex $i \in (I \cup J)$ and color class $v \in [n^\delta/w]$. We define the random variable $X_{i, v} := |N(i) \cap S_v|$ to be the number of neighbors that $u$ has in the $v$th color class. Clearly, $\E[X_{i, v}] = w \geq c \log(n)$. By the $k'$-wise independent Chernoff bound (\Cref{thm:chernoff}), we know that for $\delta' \geq 1$ and $k' \leq \delta' \lfloor \mu/e^{1/3}\rfloor$ we have $\Pr[X \geq (1+\delta')\mu] \leq e^{-k'/2}$. 
    Setting $k' = 6 \ln(n)$ and setting $\delta' = 6e^{1/3}/c \leq O(1)$, we have \[\Pr[X_{i, v} \geq (1+\delta')w] \leq 1/n^3.\] We can union bound over all the color classes and all the vertices, which tells us that each vertex's neighborhood in each color class has size at most $k := (1+\delta')w$ with probability $1-1/n$. If this does not hold for all vertices, we can just rerun the color coding.

    We will now run a second level of hashing. We pick a random hash function $h: [n] \to [k^3]$ such that each element in the image can be expressed in $\ell := 3 \log w + O(1)$ bits. We use $\ell$ copies of the linear hash function $h_r(x) := \langle r, x \rangle \bmod{2}.$ One might (rightly) worry that computing an inner product modulo 2 is \emph{not} an $\ac0$ operation. We sidestep this issue as follows: note that all of our hashing is done in the preprocessing step, and that we are only hashing $O(n^{1.5})$ values. So although the hashes technically cannot be computed in $\ac0$, we do not need to use word-level parallelism here: we can afford an extra $w$ factor in the running time of preprocessing. Note we also need at most $O(\poly(\log n))$ bits of randomness to specify the hash function.
    
    So for every vertex $i \in I \cup J$ and every color class $S_v$, we can hash the elements of $N(u) \cap S_v$, put them in sorted order, and write them down in compressed form using $k\ell/w$ words. Note that we can have collisions here, so we will store a mapping from the hashes in the image to the vertices in the preimage (written down in sorted order), which requires $O(n^{1+\delta} w)$ time. Again, we require 
    $O(\poly( \log n))$ bits of randomness, all of which  is used to build our word-packed lists in preprocessing.

    When we receive an edge query $(i, j)$, we need to determine for every color class $S_v$ whether the associated sets $N(i) \cap S_v$ and $N(j) \cap S_v$ have nonempty intersection. If we fix a color class $v$ and look at the associated packed lists, by \Cref{claim:set-intersect}, we can produce an indicator mask corresponding to their intersection in time $O(\log w)$. However, since we have hashed the vertex names, it is possible that items in the intersection are false positives. Thus, we need extract the values in the intersection, look them up in the preimage, and verify that the original entries match as well. 

    We first bound the probability of false positives. The probability that any two entries collide is $1/k^3,$ and there are $\binom{k}{2}$ pairs of entries, so we get at most $1/2k$ collisions per color class per query in expectation. We then look up the true vertex names in the preimage. To do this, we can think of these as being written down in sorted order, where each value takes up one word. We then merge the two lists and compute the true intersection in time linear in the size of the preimage, which in expectation, is $O(1/k^2).$ We can also think of this cost as the number of collisions within the neighborhoods of $i$ and $j$. So the expected cost of checking all of these collisions for a particular color class is $O(1/k),$ or linear in their number.

    Recall that we have $n^\delta/w$ color classes, so computing the indicator masks over all of them requires time $O(n^\delta \log w/w)$ per query. As a result, we can afford to spend asymptotically that same amount of time checking false positives. Let $P$ be the number of collisions we have to check. By construction, we have $\E[P] = \frac{1}{2k} \cdot \frac{n^\delta}{w} = \frac{n^\delta}{2(1+\delta')w^2}.$ By a Markov bound, 
    \[ \Pr\left[P \geq n^\delta \log w / w \right] \leq \frac{1}{2(1+\delta')w \log w} \leq \frac{1}{2}.\]
    Thus, as in the All-Edges Equality algorithm in \Cref{thm:tri-det}, if we see more than $n^\delta \log w / w$ collisions, we rerun the query in a subsequent round of the algorithm. The number of queries required is upper bounded by $n^{1+\delta} \cdot \sum_{i = 0}^\infty 2^{-i} \leq 2n^{1+\delta}.$ Thus, we require a total of $O(n^{1+\delta})$ queries, each of which takes $O(n^\delta \log w/ w)$ time. As a result, we can compute All-Edges Sparse Triangle in expected time $O(n^{1+2\delta} \log w/ w).$ 
\end{proof}

\section{Solving Other Triangle Problems Faster}
We now look at other triangle problems whose algorithms can be improved using the techniques from above. We begin by looking at  all-edges monochromatic triangle.

\subsection{All-Edges Monochromatic Triangle Detection}

Vassilevska, Williams, and Yuster gave a $O(n^{(3+\omega)/2})$ time algorithm for the detection version of the problem \cite{DBLP:journals/talg/VassilevskaWY10}. Thus, in the dense case, it seems to be meaningfully harder than triangle detection. However, in the sparse setting, we give an algorithm that runs in the same slightly subquadratic runtime as we get from \Cref{thm:all-edges} for All-Edges Sparse Triangle. We define the problem in the usual way, but require the additional stipulation of bounded degree. 

\begin{definition}[All-Edges Sparse Monochromatic Triangle]
    Let $G$ be a tripartite graph whose vertex set is $I \cup J \cup K$ and $|I| = |J| = |K| = n.$ Let the maximum degree in the graph be $O(n^\delta).$ The graph has colors along the edges; that is, there is a function $c: E \to [n^2]$. For every edge $e \in E \cap (I \times J)$, decide whether $e$ participates in a monochromatic triangle.
\end{definition}

\begin{reminder}{\Cref{thm:monochromatic}}
    There is a Las Vegas algorithm for the All-Edges Sparse Monochromatic Triangle problem that runs in expected time $O(n^{1+2\delta} \log w) / w$.
\end{reminder}

\iffullversion
\begin{proof}[Proof of \Cref{thm:monochromatic}]
    We will follow the proof of \Cref{thm:all-edges} fairly closely. We run the following slightly modified algorithm. For clarity, we will refer to the edge colors prescribed by $c$ as labels. 

    \begin{enumerate}
        \item Randomly partition $K$ into $n^\delta/\log n$ color classes using a $O(\log n)$-wise independent hash function.
        \item Compress the name of each vertex in $K$ using a linear hash function: $h: [n] \to [\poly(\ell)]$. Compress the names of the edge colors using another linear hash function $h_c: [n] \to [\poly(\ell)].$
        \item For each vertex $i$ in $I \cup J,$ for every color class $S_v$ in $[n^\delta / \log n]$, write down a compressed list that describes every node $x$ in $N(i) \cap S_v$. For every $x$, we will store its hash concatenated with the hash of the label $c(i, x).$ 
        \item To respond to an edge query $(i, j)$, for every color class, compute the set intersection of the compressed lists associated with $i$ and $j$'s neighbors with that color.  If the intersection is non-empty and the label of the two neighbors matches the label of $(i, j)$, we have found a monochromatic triangle.
    \end{enumerate}

    Let $k$ be $\max_{v \in [n^\delta/w], u \in (I \cup J)} |N(u) \cap S_v|$, the maximum number of neighbors in $K$ a given vertex in $I \cup J$ has in a particular color class. By the same argument as in the All-Edges setting, we have that after color-coding as in Step 1, $k = O(w)$ with high probability. As before, our goal is to write down matrices $A$ and $B$ of dimension $n \times \sqrt{n}$ and $\sqrt{n} \times n$ that capture $I$ and $J$'s edges to $K$, respectively. We now choose two linear hash functions $h_v: [n] \to [k^3]$ as before, and $h_c: [n^2] \to [k^3]$ for hashing the edge labels. Suppose vertex $i$ in $I$ has an edge with label $c_i$ to the $v$th vertex in the $k$th color class of $K.$ The associated hash value is $h_v(v) || h_c(c_i).$ We compute these hashes for all (at most) $k$ such vertices in $N(i) \cap S_v$, put them in sorted order, and word-pack them into $k\ell/w$ words. Again, we will have hash collisions, so we can store a mapping from the hashes in the image to the vertices in the preimage. This preprocessing (across all vertices and color classes) takes $O(n^{1.5}\log n)$ time and $O(\poly( \log n))$ random bits. 

    Thus, for a given edge query $(i, j, c(i, j))$, our task is to decide whether there exists a $v$ such that the associated sets $N(i) \cap S_v$ and $N(j) \cap S_v$ have nonzero intersection and all three labels match. We can use our compressed list operations to solve this problem slightly faster than brute force. The chunk length is $2\ell,$ since we have concatenated the hashed color class indices with hashed labels. As before, if we fix a color class and look at the associated packed lists ($N(i) \cap S_v$ and $N(j) \cap S_v$), by \Cref{claim:set-intersect}, we can produce an indicator mask that conveys their intersection in time $O(\log \log n).$ However, we need to exclude triangles that are not monochromatic.
    
    To do this, repeat the hashed color $h_c(c(i, j))$ in the lower $\ell$ bits of each of the $2\ell$-length chunks. Take the XOR of this value with every word in the packed lists $N(i) \cap S_v$ and $N(j) \cap S_v$. Then, produce an indicator mask using \Cref{claim:const-all-zeros} with chunk size $\ell$ to indicate edges with the appropriate color. We can then zero out ones in the mask corresponding to even chunk locations, as these correspond to a subset of the potential triangles and are taken care of by the earlier mask. We then shift the entire mask $\ell$ bits to the left, so that our ones are aligned along the length $2\ell$ chunks, as desired. Finally, we can AND this mask with the one created above, producing a mask that indicates all potential monochromatic triangles (but will have false positives corresponding to collisions under $h_v$ and under $h_c$).

    We now need to check false positives, as in the All-Edges version. We note that the probability that any two length-$2\ell$ entries collide (either under $h_v$ or $h_c$) is upper bounded by $2/k^3$ by union bound. As a result, by the same argument as above, we expect $O(1/k)$ collisions per color class per query in expectation. For each one, we look up both the original vertex name and also the true label, and verify that both preimages match. The expected preimage size is $O(1/k^2),$ so with high probability, we will only need to check a constant number of values here. Thus, the expected cost of checking all of these collisions is still $O(1/k).$ As before, we simply allot ourselves a ``budget'' of false positives to check, and return ``don't know'' on any queries for which we cannot check them all. By the same analysis, we obtain an expected running time of $O(n^{1+2\delta} \log w/ w)$ for this problem. 

\end{proof}
\fi

\subsection{A Faster Algorithm for Exact Triangle}

We now consider Exact Triangle, which, unlike the problems mentioned above, can be sped up even in the dense case. The problem is known to require $n^{3-o(1)}$ time under both the 3SUM hypothesis \cite{VW13} and the APSP hypothesis \cite{VW18}, but it admits a log-factor speedup using an equality product based approach similar to the one in \Cref{sec:warmup}.

\begin{definition}[Exact Triangle]
    This problem asks, given a target $T$ and an $n$ node graph with edge weights in $\{-n^c, \dots, n^c\}$ for some integer $c > 0,$ does there exist a triple of vertices $p, q, r$ such that $w(p, q) + w(q, r) + w(r, p) = T$?
\end{definition}

\begin{reminder}{\Cref{thm:exact-tri}}
    There is a Las Vegas algorithm that solves {\sc Exact Triangle} on $n$-node graphs in expected $O(n^3 \log w / w^{3/2})$ time.
\end{reminder}

\begin{proof}[Proof of \Cref{thm:exact-tri}]
    
    Without loss of generality, we may assume the graph $G = (V,E)$ is tripartite with parts $I$, $J$, and $K$, where each part has $n/3$ nodes. Given a target weight $T$, and weight function $w' : E \rightarrow {\mathbb Z}$, we want to determine if there are $i \in I$, $j \in J$, and $k \in K$ such that $i,j,k$ is a triangle and $w'(i,j) + w'(j,k) + w'(k,i) = T$. For simplicity, we assume that all weights are in $[-2^w, 2^w]$ where $w$ is the word size.

    Let $t = \sqrt{w}/(6 \sqrt{\log w})$. We split the nodes of $I$, $J$, and $K$ into at most $n/t + 1$ \emph{groups} of $t$ nodes each. Our idea is to check each triple of groups (containing $3t$ nodes and $O(t^2)$ edges) for an exact triangle, using a constant number of word operations.  

    To this end, we let $\ell := 10 \log w$, and we will hash all edge weights modulo a random prime $p$ from the integer interval $\left[2, 2^\ell\right]$. (To be clear, we replace each edge weight with the corresponding residue in $\{0,1,\ldots,p-1\}$.) By the Prime Number Theorem, a random integer in this range is prime with probability $\Omega(1/\ell)$. We can efficiently find such a prime by picking integers at random and primality testing. Note that this process costs only an (additive) $\poly(\log w)$ factor. 

    Now take some triple of groups $I' \subseteq I$, $J' \subseteq J$, $K' \subseteq K$,  where $|I'|=|J'|=|K'|=t$. The number of bits needed to encode the weighted subgraph induced by $I' \cup J' \cup K'$ is
    \[3t^2 \cdot \ell = 3\left(\frac{\sqrt{w}}{6 \sqrt{\log w}}\right)^2 \cdot (10 \log w) < w,\] so the entire subgraph can be packed into a single word. 

    Next, we observe that the problem of checking whether a given triple $I',J',K'$ contains an exact triangle modulo $p$ can be computed with a $\poly(w)$-size $\AC0$ circuit. Let $T_p = (T \bmod p)$ be the corresponding integer in $\{0,1,\ldots,p-1\}$. Our circuit takes an OR over all $O(t^3) \leq O(w^{3/2})$ triples of nodes $i \in I'$, $j\in J'$, and $k \in K'$, then checks that the edges $(i,j),(j,k), (k,i)$ exist and \[w'(i,j) + w'(j,k) + w'(k,i) \in \{T_p, p + T_p, 2p + T_p\}.\] (Since only three weights are being summed, these are all the possible choices for the integer sum to equal the target modulo $p$.) Since addition is in $\AC0$, this construction is computable in $\AC0$. As a consequence, there is a single word operation to determine whether the triple $I',J',K'$ contains an exact triangle mod $p$. For notational brevity, let's call this an \emph{mod triangle check} which returns a yes or no answer.

    All of the ingredients are now in place for us to describe the algorithm in full. Pick a uniform random prime $p$ from $\left[2, 2^\ell\right]$, hash all edge weights modulo $p$, and for all $(n/t)^3$ triples $I', J', K'$ on $3t$ nodes, we run a mod triangle check in constant time. For each check that returns yes, we enumerate all $O(t^3)$ triples of nodes $i \in I'$, $j \in J'$, $k \in K'$ to determine if there is an exact triangle in $I',J', K'$ with the original unhashed weights (returning this triangle if it's found). If more than $(n/(tw))^3$ checks of triples return yes, then we stop and say \emph{don't know}. Otherwise, we return \emph{reject} if no exact triangle was found after enumerating over all triples $I',J',K'$. 

    First we verify the running time. By construction, our algorithm takes time $\poly(\log w) + O(n^2)$ to pick a random prime of $\ell=\Theta(\log w)$ bits and hash all weights to $\ell$ bits. Next, it takes $O((n/t)^3)$ time to call the mod triangle check on all triples $I',J',K'$. Finally, it takes $O((n/(tw))^3 \cdot t^3) \leq O((n/w)^3)$ extra time to go through up to $(n/(tw))^3$ triples and exhaustively check the corresponding $O(t)$-node subgraph for an exact triangle. Since $t < w$, the running time is bottlenecked by $O(n^3/t^3) \leq O(n^3 (\log w)^{3/2}/w^{3/2})$.

    Next, we show correctness. Note that if there is an exact triangle solution, then it is in some triple $I'$, $J'$, $K'$, and the mod triangle check called on $I'$, $J'$, $K'$ will always return \emph{yes} no matter what prime is chosen. Thus the algorithm will find the exact triangle, \emph{provided the algorithm does not stop early and return} \emph{don't know}. Moreover, if the algorithm does not return \emph{don't know} when there is no exact triangle, then the algorithm correctly rejects.
    
    Finally, similarly to \Cref{thm:tri-det}, we prove that the algorithm outputs \emph{don't know} with probability at most $1/\poly(w)$. Hence with probability at least $1-1/\poly(w)$, the algorithm does not stop early: it will return an exact triangle when one exists, and will correctly reject when an exact triangle does not exist. 

    Fix some triple $I', J', K'$ for which we will run a mod triangle check. For each $i \in I'$, $j \in J'$, and $k \in K'$, the probability that \[w'(i,j) + w'(j,k) + w'(k,i) \neq t ~\text{ and }~ w'(i,j) + w'(j,k) + w'(k,i) = t \bmod p\] is at most  the number of prime factors of $t-w'(i,j) + w'(j,k) + w'(k,i)$ divided by the total number of primes in $[2,2^{\ell}]$, which is at most $2w \cdot (\ell/2^{\ell}) \leq (20 \log w)/w^9$. (Here, we apply the assumption that each edge weight fits in a word, in which case each weight has at most $w$ prime factors, and their sum has at most $2w$ prime factors.) By a union bound over all $t^3$ triples of nodes in $I',J',K'$, the probability that some $i,j,k \in I',J',K'$ do not form an exact triangle yet the mod triangle check returns yes on $I',J',K'$ is at most $20t^3 (\log w)/w^9$. As there are $(n/t)^3$ triples in total, the \emph{expected} number of false-positive triples (there is no exact triangle yet the mod triangle check returns yes) is at most $20n^3 (\log w)/w^9$. By Markov's inequality, the probability that the number of false-positive triples exceeds $n^3/(t^3 w^3)$ is at most \[20t^3(\log w)/w^6 < 20w^{3/2}(\log w)/w^6 < 20/w^4\] for $w$ sufficiently large. Therefore the probability the algorithm outputs \emph{don't know} is at most $20/w^4$, which is $o(1)$.
\end{proof}

\section{Attacks on 4-Cycle}
We now turn to the problem of 4-cycle detection. Our high-level approach to reduce to a problem called Concatenated Element Distinctness, and then speed up that problem via packing into words. We assume without loss of generality that the graph is bipartite: by imposing a random bipartition and deleting edges that violate it, any 4-cycle is preserved with constant nonzero probability. We now show that it suffices to solve 4-cycle in graphs with maximum degree $O(\sqrt{n}).$ 

\begin{lemma} \label{lemma:few-high-deg}
    Suppose a bipartite graph $G = (L \cup R, E)$ on $n > 2$ nodes is 4-cycle free. Then for all $d \geq \sqrt{2n}$, $G$ has at most $2n/d$ nodes with degree at least $d$.
\end{lemma}

\iffullversion
\begin{proof}[Proof of \Cref{lemma:few-high-deg}]
    Suppose there exist $N := 2n/d + 1$ nodes with degree at least $d,$ which we will call $v_1, \dots, v_N.$ Because $G$ is 4-cycle free by assumption, each of the neighborhoods of these nodes have intersection of size at most 1. $v_1$ has $d$ neighbors. The number of neighbors of $v_2$ that are not neighbors of $v_1$ is at least $d - 1$, and the neighbors of $v_3$ that are not neighbors of either $v_1$ or $v_2$ is at least $d-2.$ So the total number of distinct neighbors of these nodes is
    \[ d + (d-1) + \dots + (d-N+1)
        = dN - \sum_{i=1}^{N-1} i 
        = d\left(\frac{2n}{d}+1\right) - \frac{N(N-1)}{2} = 2n \left(1 - \frac{n}{d^2}\right) + \left(d - \frac{n}{d}\right). \]

    Observe that if we let $d \geq \sqrt{2n}$, this expression is at least $n + 1$ for $n > 2$, which is a contradiction.
\end{proof}
\fi

Thus, the number of high-degree nodes is relatively small, which allows us to brute-force over them in subquadratic time. In particular, we can run the following preprocessing step.

\begin{lemma} \label{lemma:c4-max-deg}
    Given a bipartite graph $G = (V, E)$, there is a $O(n^{3/2})$ preprocessing algorithm that either returns ``4-cycle'' or produces a graph $G'$ with maximum degree $\Delta = \sqrt{2n}$ such that solving 4-cycle on $G$ and $G'$ gives the same answer.
\end{lemma} 

\iffullversion
\begin{proof}[Proof of \Cref{lemma:c4-max-deg}]
    If $|E| \geq \frac{n}{4}(1+\sqrt{4n-3})$ or if more than $\sqrt{2n}$ vertices with degree at least $\sqrt{2n}$, return ``4-cycle'' (by Reiman's theorem and \Cref{lemma:few-high-deg}, respectively). Otherwise, let $N \leq \sqrt{2n}$ be the number of vertices that have degree at least $\sqrt{2n}.$ From each of these vertices, run a modified version of breadth-first search that terminates when any node is visited twice (this corresponds to cycle detection). This requires $O(n)$ time per node, for a total of $O(n^{3/2})$ time. If any of these breadth-first searches terminate after running three or fewer ``steps'' outward, we have found a cycle of length at most 5. Since the graph is bipartite by assumption, it cannot have odd cycles, so we can return ``4-cycle.'' After verifying that none of the nodes in $N$ participate in a 4-cycle, we can return $G',$ the subgraph of $G$ induced by $V \setminus N.$ Both the correctness and the degree bound follow directly from the construction. 
\end{proof}
\fi

Because of this, it suffices to solve 4-cycle in the following setting:

\begin{definition}[4-Cycle Detection]
    Given a graph with $m = O(n^{3/2})$ edges and maximum degree $\Delta = O(\sqrt{n})$, decide whether there exists a 4-cycle.
\end{definition}


\begin{definition}[Concatenated Element Distinctness]
    Given $k$ column vectors $L_1, \dots, L_k$ of $n$ $\ell$-bit elements each, for every pair $i, j$ of columns, element-wise concatenate them to a produce a new list $L_{i, j}$ containing $n$ length-$2\ell$ elements. If for all $i, j$, the list $L_{i, j}$ has distinct elements, accept, else reject.
\end{definition}

We now reduce 4-cycle detection to this problem, and then use the word tricks from \Cref{sec:wordtricks} on Concatenated Element Distinctness to improve the 4-cycle runtime when we have large word size.


\begin{reminder}{\Cref{thm:reduction}}
    For $T(n,k,\ell) \geq n+k+\ell$, if Concatenated Element Distinctness can be solved in $O(T(n,k,\ell))$ time, then 4-cycle detection in $n$-node graphs can be done in $O(T(n,\sqrt{n},(\log_2 n)/2))$ time. This reduction succeeds with constant probability.
\end{reminder}

\iffullversion
\begin{proof}[Proof of \Cref{thm:reduction}]
    We begin by imposing a bipartition on the graph, randomly assigning half the vertices to each side and deleting all edges within each side. The probability that all edges in a particular 4-cycle survive is constant. We then use color coding to reduce to an equality product-like problem, similarly to the sparse triangle reductions.

    In particular, let $(L \cup R, E)$ be the new graph. Randomly partition $R$ into $\sqrt{n}$ color classes, and for every vertex $v \in L,$ if $v$ has multiple neighbors in a particular color class, delete all but one of the associated edges. As before, every edge in the graph survives with constant probability. Once we have this more structured instance, we can write column vectors corresponding to each color class. 

    For every color class $k \in [\sqrt{n}]$, write a column vector $L_k$ such that $L_k[i]$ is the element of the $k$th color class that $i \in L$ has an edge to (or $\bot$ if there is no such edge). Now, we claim that it suffices to solve the Concatenated Element Distinctness problem to solve Bounded-Degree Four Cycle Detection. To see this, if we have a yes-instance of the former, there are columns $i$ and $j$ such that if we concatenate them elementwise, the resulting column has duplicate values, say at indices $u$ and $v$. By construction, this implies that $L_i[u] = L_i[v]$ and $L_j[u] = L_j[v]$, so $u$ and $v$ share a neighbor in the $i$th color class, namely $L_i[u],$ and also in the $j$th color class, $L_j[u]$. Thus, these four nodes form a 4-cycle. The reverse direction is symmetric.

    Each of the steps here take linear time in the size of the graph, which here is $O(m + n) = O(n^{3/2})$ time. We note that there are a total of $k := \sqrt{n}$ columns, corresponding to the color classes, and that the indices of each node within a color class requires $\ell := \log_2 n /2$ bits to write down as they are over the universe $[\sqrt{n}].$
\end{proof}
\fi


\begin{lemma} \label{lemma:concateltdist}
    There is a deterministic algorithm for Concatenated Element Distinctness, on $O(\sqrt{n})$ lists each containing $n$ elements of size $O(\log n)$ that runs in $O(n^2 \log^2 n / w)$ time.
\end{lemma}

\iffullversion
\begin{proof}[Proof of \Cref{lemma:concateltdist}]
     We begin by preprocessing the input lists. Let $\ell$ be the maximum length of any element. We can think of each word as being composed of size-$2\ell$ blocks. We write down two copies of each input list, one with every element stored in the upper $\ell$ bits of the blocks, with the lower $\ell$ bits set to zeroes, and another with every element stored in the lower $\ell$ bits, with the upper $\ell$ bits set to zeroes. Note that each list requires $O(n \log n / w)$ words to write down. 

    To concatenate two column vectors $L_i$ and $L_j$ elementwise, it suffices to OR the upper-aligned version of $L_i$ with the lower-aligned version of $L_j$. This gives us a word-packed version of the concatenated list, taking up $O(n\log n/w)$ words. We then run the word-packed sorting algorithm in \Cref{claim:log2sort} with block size $2\ell,$ which takes $O(n\log^2 n/w))$ time. 

    Once we have a sorted, word-packed list, we can detect duplicate entries in linear time in the number of words. Clearly, duplicates must be adjacent entries in the list. Thus, we can simply make a second copy of the list, right-shift it by the block size $2\ell$, and XOR the two copies. As before, it suffices to find any all-zeroes block, which we can do in constant time per word using \Cref{claim:const-all-zeros}.

    We then run this procedure for all pairs of columns, which takes $O(n^2 \log^2 n /w )$ time.
\end{proof}
\fi

\subsection{Subquadratic 4-Cycle Finding in Very Dense Graphs}

Recall that by \Cref{thm:reiman}, a 4-cycle free graph has at most $\frac{n}{4}\left(1+\sqrt{4n-3}\right) = O(n^{3/2})$ edges. Here, we give a sublinear time (in the number of edges) algorithm that finds a 4-cycle when the graph is noticeably denser, i.e., when the number of edges is $m \geq n^{3/2+\delta}$. Recall that a \emph{$k$-core} of $G$ is a maximal subgraph $G'$ such that every vertex has degree at least $k$.  

\begin{lemma} \label{lemma:sqrtn-core}
    For all sufficiently large $n$, every graph $G$ with $m = n^{3/2 + \delta}$ edges has a $(\sqrt{n}+1)$-core containing at least $\Omega(\sqrt{m})$ vertices.
\end{lemma}

\iffullversion
\begin{proof}[Proof of \Cref{lemma:sqrtn-core}]
    Recall that we may obtain a $k$-core by repeatedly deleting all vertices with degree smaller than $k$ and all of their incident edges, until all remaining vertices have degree at least $k$. Suppose that the $k := \sqrt{n}+1$ core obtained by this procedure has $s$ vertices. Then we will recover it after $n-s$ rounds of peeling, deleting at most $k(n-s)$ edges at each step. It follows that $G'$ has at least $m- k(n-s)$ edges. Since $G'$ has $s$ vertices, $\binom{s}{2}$ is an upper bound on its number of edges. Therefore, we have that 
    \begin{align*}
    m - k(n-s) &\le \frac{s(s-1)}{2}  \\
    0 &\le s^2 + (-2k-1)s + 2nk - 2m
    \intertext{We are interested in values of $s$ larger than its positive root, so}
    s &\ge \frac{1}{2} \left((2k+1) + \sqrt{4k^2 + 4k + 1 - 8nk + 8m}\right) \\
      &= \sqrt{n} + \frac{3}{2} 
      + \frac{1}{2}\sqrt{4n + 12\sqrt{n} + 9 
        - 8n^{3/2} - 8n + 8n^{3/2+\delta}} \\
      &= \sqrt{n} + \frac{3}{2} 
      + \frac{1}{2}\sqrt{8n^{3/2+\delta} 
        - 8n^{3/2} - 4n + 12\sqrt{n} + 9} \\
      &\ge \Omega(n^{3/4 + \delta/2}).
\end{align*}
\end{proof}
\fi

We will use this fact to provide a sublinear (in the size of the graph) time algorithm for finding a 4-cycle when the graph is extremely dense. 

\begin{reminder}{\Cref{thm:sublinear}}
There is a Las Vegas algorithm that, given a graph with $n$ vertices and $n^{3/2 + \delta}$ edges for $\delta \in (0, \frac{1}{2})$, returns a 4-cycle and runs in expected $O(n^{5/4-\delta/2})$ time. 
\end{reminder}

\iffullversion
\begin{proof}[Proof of \Cref{thm:sublinear}]
For simplicity, let the vertex set $V = [n]$. We assume access to two oracles with information about the graph. The first takes queries of the form $(v, i) \in [n] \times [n]$, and outputs either the $i$-th neighbor of $v \in V$, or $\bot$ if $i > \deg(v)$. The second is a degree oracle that takes $v \in V$ as input and outputs $\deg(v)$. We repeat the following procedure until we find a 4-cycle: 
    \begin{enumerate}
        \item Sample a vertex $v \gets V$ uniformly at random.
        \item Make at most $n-1$ queries to the neighbor oracle to construct a list of all of $v$'s neighbors. Query the degree oracle for each neighbor.
        \item If $v$ has at least $\sqrt{n}+1$ neighbors, each of which also has degree at least $\sqrt{n}+1$, then by the pigeonhole principle, $v$ participates in a 4-cycle with two of these neighbors (the total number of distinct $2$-paths to $v$ is at least $(\sqrt{n}+1)\cdot \sqrt{n} > n$, so some node has two distinct $2$-paths to $v$). Enumerate the first $\sqrt{n}+1$ neighbors of each of the neighbors of $v$ by querying the neighbor oracle, and output the 4-cycle once it is found.
        \item If this condition does not hold, repeat.
    \end{enumerate}

    We claim that if we sample a vertex $v$ in the $(\sqrt{n}+1)$-core of $G$, we will satisfy the condition in step 3. This is because by definition, $v$ is in a subgraph such that all vertices have at least $\sqrt{n}+1$ neighbors. Since the size of the $(\sqrt{n}+1)$-core is at least $\Omega(n^{3/4 + \delta/2})$, we expect to sample a vertex in the core in $O(n^{1/4-\delta/2})$ iterations. Each iteration takes $O(n)$ time, so our overall expected runtime is $O(n^{5/4-\delta/2})$.
\end{proof}
\fi

\section{Conclusion}

In this paper, we have shown how a variety of ideas can be applied to achieve faster algorithms for some tricky triangle problems at the heart of fine-grained complexity. We conclude with some of the most significant questions left open by our work.

\begin{itemize}
    \item We have shaved $\log(n)$ and word factors from the running times of sparse graph problems, assuming a bound on the maximum degree. Can we still improve the running times of these algorithms in the more general case of sparse graphs without a degree bound? For example, is there an algorithm for {\sc All-Edges Sparse Triangle} running in $o(m^{4/3})$ time on $m$-edge graphs, assuming optimal matrix multiplication? Our approach requires at least two of the nodes in the triangle to have low degree, which is not enough to apply low-degree/high-degree tricks such as those in \cite{AYZ97}.
    \item Are there \emph{deterministic} algorithms for these triangle problems that have similar running times? All of our algorithms rely on the use of randomness. Note that since all of our algorithms use $O(\log n)$-wise independent hash functions and linear hashing for small $k$, we only require at most $\poly(\log n)$ bits of randomness. 
    \item The {\sc All-Edges Sparse Triangle} problem cannot be solved much faster than the obvious running time, under the 3SUM and APSP hardness hypotheses. Is there a fine-grained reduction from {\sc All-Edges Sparse Triangle} to {\sc Sparse Triangle Detection}, so that detecting a triangle faster on sparse (or low-degree) graphs would yield a faster algorithm for the all-edges case? Such a reduction is well-known in the dense case~\cite{VW18}.

    \item Can 4-cycle detection be solved in $O(n^2 \log n)/w$ time on the Word RAM for $w \gg \Omega(\log n)$? Is there any reason to believe that such an algorithm may not exist? 
\end{itemize}

\bibliography{refs}

\end{document}